\documentclass[twoside,leqno,twocolumn]{article}

\usepackage[letterpaper]{geometry}

\usepackage{ltexpprt}
\usepackage[hidelinks]{hyperref}

\usepackage{graphicx}

\usepackage{amsmath,amssymb,amsfonts}
\newtheorem{definition}{Definition}

\usepackage{booktabs}
\usepackage{multirow}
\usepackage{hhline}
\usepackage{subcaption}
\usepackage{cleveref}
\usepackage{enumitem}
\usepackage{paralist}

\usepackage{xcolor,colortbl}
\definecolor{Gray}{gray}{0.85}
\newcolumntype{a}{>{\columncolor{Gray}}c}

\usepackage[algo2e,linesnumbered,ruled,vlined]{algorithm2e}
\DontPrintSemicolon
\SetKwComment{note}{$\triangleright$ }{}
\SetFuncSty{textsc}
\SetCommentSty{textit}
\SetDataSty{texttt}
\SetKwInput{Initial}{Initial}
\SetKwInput{Parameter}{Param.}
\SetKwInput{Input}{Input}
\SetKwInput{Data}{Data}
\SetKwInput{Output}{Output}
\SetKw{ForParallel}{parallel}
\ResetInOut{Result} 

\SetCommentSty{mycommfont}

\newcommand{\sz}{\scriptstyle }
\newcommand{\tg}{\mathcal{G}\xspace}
\newcommand{\tge}{\mathcal{E}\xspace}
\newcommand{\edgestream}{\xi\xspace}
\newtheorem{observation}{Observation}
\usepackage{tikz}
\usetikzlibrary{arrows,decorations.pathmorphing}
\usetikzlibrary{arrows.meta}

\begin{document}

\title{\Large An Index for Single Source All Destinations Distance Queries in Temporal Graphs}

\author{Lutz Oettershagen\thanks{{Institute of Computer Science},
        {University of Bonn},
        \texttt{\{lutz.oettershagen,petra.mutzel\}@cs.uni-bonn.de}}
\and Petra Mutzel\textsuperscript{*}
}

\date{}

\maketitle

\begin{abstract} \small\baselineskip=9pt 
Temporal closeness is a generalization of the classical closeness centrality measure for analyzing evolving networks. 
The temporal closeness of a vertex $v$ is defined as the sum of the reciprocals of the temporal distances to the other vertices. 
Ranking all vertices of a network according to the temporal closeness is computationally expensive as it leads to
a single-source-all-destination (SSAD) temporal distance query starting from each vertex of the graph.
To reduce the running time of temporal closeness computations, we introduce an index to speed up SSAD temporal distance queries called \emph{Substream} index. We show that deciding if a \emph{Substream} index of a given size exists is NP-complete and provide an efficient greedy approximation. Moreover, we improve the running time of the approximation using min-hashing and parallelization. Our evaluation with real-world temporal networks shows a running time improvement of up to one order of magnitude compared to the state-of-the-art temporal closeness ranking algorithms.
\end{abstract}

\section{Introduction}

Computing closeness centrality is an essential task in network analysis and data mining~\cite{braha2009time,oettershagen2020efficient,santoro2011time,tang2010analysing}.
In this work, we focus on improving the efficiency of computing temporal closeness in evolving networks.
We represent evolving networks using \emph{temporal graphs}, which consists of a finite set of vertices and a finite set of temporal edges. Each temporal edge is only available at a specific discrete point in time, and edge transition takes a strictly positive amount of time.
Temporal graphs are often good models for real-life scenarios due to the inherently dynamic nature of most real-world activities and processes.
For example, temporal graphs are used to model and analyze bioinformatics networks~\cite{lebre2010statistical,przytycka2010toward}, communication networks~\cite{candia2008uncovering,eckmann2004entropy}, contact networks~\cite{ciaperoni2020relevance,oettershagen2020classifying}, social networks~\cite{holme2004structure,moinet2015burstiness}, and transportation networks~\cite{pyrga2008efficient}. 

Temporal closeness is one of the popular and essential centrality measures for temporal networks, and various variants of temporal closeness have been discussed~\cite{magnien2015time,pan2011path,tang2010analysing,santoro2011time}. Here, we consider one of the standard variants, the \emph{harmonic temporal closeness} of a vertex, which is defined as the sum of the reciprocals of the \emph{durations} of the fastest paths to all other vertices~\cite{oettershagen2020efficient}. 
Unfortunately, the computation with respect to the minimum duration distance is expensive and can be prohibitive for large temporal networks~\cite{oettershagen2020efficient,wu2014path}.
To overcome this obstacle, we introduce an index for efficiently answering single-source-all-destination (SSAD) temporal distance queries in temporal graphs.

\noindent\textbf{Our idea:}
We exploit the often very limited reachability in temporal graphs.
\begin{figure}
    \centering
    \begin{subfigure}{0.3\linewidth}
        \centering
        \begin{tikzpicture}[scale=0.8]
            \node[style={}] (0,-1.5){};
            \begin{scope}[every node/.style={circle,thick,draw,minimum size=4mm,inner sep=0pt}]
                \node (A) at (0,3) {$a$};
                \node (B) at (0,1.5) {$c$};
                \node (C) at (1.5,3) {$b$};
                \node (D) at (1.5,1.5) {$d$};
                \node (E) at (0,0) {$e$};
                \node (F) at (1.5,0) {$f$} ;
            \end{scope}
            
            \begin{scope}
                \path [->] (A) edge[color=red,thick,bend left=15] node[right] {$3$} (B);
                \path [->] (A) edge[color=red,thick,bend left=15] node[above] {$1$} (C);
                \path [->] (A) edge[color=red,thick,bend right=15] node[below] {$2$} (C);
                \path [->] (B) edge[thick,bend left=15] node[left] {$2$} (A);
                \path [->] (B) edge[color=red,thick] node[left] {$9$} (E);
                \path [->] (C) edge[color=red,thick] node[right] {$3$} (D);
                \path [->] (D) edge[thick] node[right] {$1$} (F);
                \path [->] (E) edge[thick] node[above] {$6$} (F);
                \path [->] (F) edge[thick] node[above] {$7$} (B);
            \end{scope}
        \end{tikzpicture}
        \label{fig:examplea}
        \caption{$\tg$}
    \end{subfigure}\hfill
    \begin{subfigure}{0.3\linewidth}
        \centering
        \begin{tikzpicture}[scale=0.8]
            \begin{scope}[every node/.style={circle,thick,draw,minimum size=4mm,inner sep=0pt}]
                \node[color=blue] (A) at (0,3) {$a$};
                \node[color=blue] (B) at (0,1.5) {$c$};
                \node[color=blue] (C) at (1.5,3) {$b$};
                \node (D) at (1.5,1.5) {$d$};
                \node (E) at (0,0) {$e$};
                \node[gray] (F) at (1.5,0) {$f$} ;
            \end{scope}
            
            \begin{scope}
                \path [->] (A) edge[thick,bend left=15] node[right] {$3$} (B);
                \path [->] (A) edge[thick,bend left=15] node[above] {$1$} (C);
                \path [->] (A) edge[thick,bend right=15] node[below] {$2$} (C);
                \path [->] (B) edge[thick,bend left=15] node[left] {$2$} (A);
                \path [->] (B) edge[thick] node[left] {$9$} (E);
                \path [->] (C) edge[thick] node[right] {$3$} (D);
            \end{scope}
        \end{tikzpicture}
        \caption{$\tg_1$}
        \label{fig:exampleb}
    \end{subfigure}\hfill
    \begin{subfigure}{0.3\linewidth}
        \centering
        \begin{tikzpicture}[scale=0.8]
            \node[style={}] (x) at (0,3.4){ };
            \begin{scope}[every node/.style={circle,thick,draw,minimum size=4mm,inner sep=0pt}]
                \node[gray] (A) at (0,3) {$a$};
                \node (B) at (0,1.5) {$c$};
                \node[gray] (C) at (1.5,3) {$b$};
                \node[color=blue] (D) at (1.5,1.5) {$d$};
                \node[color=blue] (E) at (0,0) {$e$};
                \node[color=blue] (F) at (1.5,0) {$f$} ;
            \end{scope}
            
            \begin{scope}
                \path [->] (B) edge[thick] node[left] {$9$} (E);
                \path [->] (D) edge[thick] node[right] {$1$} (F);
                \path [->] (E) edge[thick] node[above] {$6$} (F);
                \path [->] (F) edge[thick] node[above] {$7$} (B);
            \end{scope}
        \end{tikzpicture}
        \caption{$\tg_2$}
    \end{subfigure}
    \caption{(a) Temporal graph $\tg$ with availability times shown at the edges. The transition times are one for all edges. 
        The edge stream $\edgestream(a)$ is highlighted red; it contains all edges that can be part of a temporal walk starting at $a$.
        (b) Temporal subgraph $\tg_1$ of $\tg$ contains all edges, such that the temporal closeness of the vertices $a$, $b$, and $c$ can be computed. (c) Temporal subgraph $\tg_2$ contains all edges, such that the temporal closeness for $d$, $e$, or $f$ can be determined.
        Vertices that can not be reached are colored gray.\vspace{-2mm}
    }
    \label{fig:example}
\end{figure}
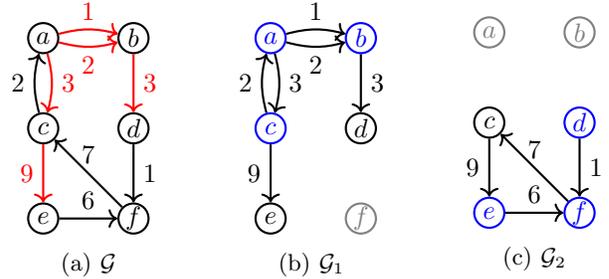
Given a temporal graph, we construct $k$ smaller (possibly non-disjoint) subgraphs. 
We guarantee that for each vertex of the input graph, there exists one of the $k$ subgraphs, such that temporal distance queries, and hence its temporal closeness, can be answered with a single pass over the chronologically ordered edges of the subgraph using a state-of-the-art streaming algorithm introduced in~\cite{wu2014path}. 
For example, \Cref{fig:example} shows a temporal graph (a), for which the temporal distance queries starting from vertices $a$, $b$, or $c$ can be answered using only the temporal subgraph shown in (b) and starting from any of the remaining vertices using only the temporal subgraph shown in (c).  

It is insufficient to store the temporal graph in an adjacency list representation and compute the minimum duration distance using
label setting algorithms. The streaming approach for computing the temporal distances is often already significantly faster~\cite{oettershagen2020efficient,wu2014path}.  

\smallskip\noindent\textbf{Contributions:} 
\begin{compactenum}
    \item We propose the \emph{Substream} index for temporal closeness computation in temporal graphs. 
    We show that deciding if a \emph{Substream} index of a given size can be constructed is an NP-complete problem. 
    \item We introduce an efficient approximation for constructing a \emph{Substream} with guarantees on the resulting index size.
    Next, we improve our approximation with min-hashing and shared-memory parallelization to speed up the index construction.
    \item In our evaluation on real-world temporal graphs, we show that our approach achieves up to an order of magnitude faster temporal closeness computation times (indexing + querying) compared to the state-of-the-art algorithms.
\end{compactenum}

\section{Preliminaries}\label{sec:prelim}
We use $\mathbb{N}$ to denote the strictly positive integers.
For $\ell\in\mathbb{N}$, we denote with $[\ell]$ the set $\{1,\ldots,\ell\}$.
A \emph{directed temporal graph} $\tg=(V, \tge)$ consists of a finite set of vertices $V$ and a finite set $\tge$ of directed \emph{temporal edges} $e=(u,v,t,\lambda)$ with \emph{tail} $u$ and \emph{head} $v$ in $V$, \emph{availability time} (or \emph{timestamp}) $t \in \mathbb{N}$ and \emph{transition time} $\lambda \in \mathbb{N}$. The transition time of an edge denotes the time required to traverse the edge.
We only consider directed temporal graphs---it is possible to model undirectedness by using a forward- and a backward-directed edge with equal timestamps and transition times for each undirected edge.
We call a sequence $S=(e_1,\ldots,e_m)$ of $m$ temporal edges with non-decreasing availability times (ties are broken arbitrarily) a \emph{temporal edge stream}.
A temporal edge stream $S$ induces a temporal graph $\tg=(V(S),S)$ with $V(S)=\{u,v\mid (u,v,t,\lambda)\in S \}$, where we, for notational convenience, interpret the sequence $S$ as a set of edges. Given a temporal edge stream $S$, we denote with $n$ the number of vertices of the induced temporal graph,
and with $n^+$ the number of vertices with at least one outgoing edge, i.e., non-sink vertices. 
Let $S$ and $S'$ be temporal edge streams and $S'\subseteq S$, i.e., $S'$ contains only edges from $S$. We call $S'$ a \emph{substream} of $S$.
If it is clear from the context, we use the view of a temporal graph $\tg$ and the corresponding edge stream interchangeably.
The size of an edge stream $S$ consisting of $m$ edges is $|S|=m$. 
We assume that $m\geq \frac{n}{2}$, which holds unless isolated vertices exist, to simplify the discussion of running time complexities. Note that edge streams do not have isolated vertices.
Let $S_1$ and $S_2$ be two temporal edge streams, we denote by $S_3=S_1\cup S_2$ the union of the two temporal edges streams and $S_3$ is a temporal edge stream, i.e., the edges of $S_3$ are ordered in non-decreasing order of their availability time. 
$S_3$ can be computed in $\mathcal{O}(|S_1|+|S_2|)$ time due to the ordering of the edges in non-decreasing availability times.
We denote the lifetime of a temporal graph $\tg=(V,\tge)$ (or edge stream $S$) with $T(\tg)=[t_{min},t_{max}]$ with $t_{min}=\min\{t\mid e=(u,v,t,\lambda)\in \tge\}$ and $t_{max}=\max\{t+\lambda\mid e=(u,v,t,\lambda)\in \tge\}$.

\smallskip
\noindent
\textbf{Temporal Distance and Closeness}\label{sec:temporalpath}
A \emph{temporal walk} between vertices $v_1,v_{\ell+1}\in V$ of length $\ell$ is a sequence of $\ell$ temporal edges 
$(e_1=(v_1,v_2,t_1,\lambda_1),e_2=(v_2,v_3,t_2,\lambda_2),$ $\dots,e_\ell=(v_\ell,v_{\ell+1},t_\ell,\lambda_\ell))$
such the head of $e_i$ equals the tail of $e_{i+1}$, and $t_{i}+\lambda_i \leq t_{i+1}$ for $1\leq i <\ell$. 
A \emph{temporal path} $P$ is a temporal walk in which each vertex is visited at most once.
The \emph{starting time} of $P$ is $s(P)=t_1$, the \emph{arrival time} is $a(P)=t_\ell+\lambda_\ell$, and
the \emph{duration} is $d(P)=a(P)-s(P)$. 
A \emph{minimum duration} or \emph{fastest} $(u,v)$-path is a path from $u$ to $v$ with shortest duration among all paths from $u$ to $v$.
The harmonic temporal closeness is defined in terms of minimum duration.
\begin{definition}[Harmonic Temporal Closeness]\label{tcc:def:closeness}
    Let $\tg=(V,\tge)$ be a temporal graph.
    We define the harmonic temporal closeness for $u\in V$ as
    $ c(u) = \sum_{v\in V\setminus \{u\}}\frac{1}{d(u,v)}$. 
\end{definition}
If $v$ is not reachable from vertex $u$, we set $d(u,v)=\infty$, and we define $\frac{1}{\infty}=0$.

\smallskip
\noindent
\textbf{Temporal Reachability}
We say vertex $v$ is \emph{reachable} from vertex $u$ if there exists a temporal $(u,v)$-path. 
We denote with $\edgestream(v)$ the subset of edges that can be used by any temporal walk starting at $v$ ordered in non-decreasing availability times (ties are broken arbitrarily), i.e., $\edgestream(v)$ is a temporal edge stream.
For example, in \Cref{fig:example}a the edges of $\edgestream(a)$ are highlighted in red.
Temporal graphs are, in general, not strongly connected and have limited reachability with respect to temporal paths due to the missing symmetry and transitivity.

\smallskip
\noindent
\textbf{Restrictive Interval}
Temporal distance, reachability, and closeness computations can be additionally restricted to a time interval $\tau=[a,b]$ such that only edges $e=(u,v,t,\lambda)\in\tge$ are considered that start and arrive in the interval $\tau$, i.e., $a \leq t < b$ and $a < t+\lambda \leq b$.

\section{Substream Index}\label{sec:substream}
The substream index constructs $k$ temporal subgraphs from a given temporal graph, leveraging the following simple observation.
\begin{observation}\label{obs:reach}
The temporal durations between vertex $v\in V$ and all other vertices, can be determined solely with edges in $\edgestream(v)$.
\end{observation}

Given a temporal graph $\tg=(V,\tge)$ in its edge stream representation $S$ and $2\leq k < n$, we first construct $k$ substreams $S_1,\ldots,S_k$ of $S$. 
Each vertex $v\in V$ is assigned to exactly one of the new substreams, such that the substream contains all edges that can be used by any temporal walk leaving $v$.
We use an additional empty stream $S_0=\emptyset$ to which we assign all sink-vertices, i.e., vertices with no outgoing edges.
The substreams and the vertex assignment together form the substream index, which we define as follows.
\begin{definition}
Let $S$ be a temporal edge stream and $k\geq 2\in\mathbb{N}$. 
We define the pair $\mathcal{I}=(\mathcal{S}, f)$ with $\mathcal{S}=\{S_0,S_1,\ldots,S_k\}$, $S_i\subseteq S$ for $1\leq i\leq k$, $S_0=\emptyset$, and $f:V\rightarrow [k]\cup \{0\}$ as \emph{substream index}, 
with $f$ maps $v\in V(S)$ to an index $i$ of subset $S_i\in \mathcal{S}$, such that $\edgestream(v)\subseteq S_i$.
\end{definition}
A pair $(v,\tau)$, with $v\in V$ and $\tau$ is a restrictive time interval is a query to the substream index.
We answer it by running the fastest paths streaming algorithm from~\cite{wu2014path}, on the substream $S_{f(v)}$ for vertex $v$ and restricting time interval $\tau$. 
The streaming algorithm uses a single pass over the edges in the substream $S_{f(v)}$.
Next, we define the size of a substream index as the maximum substream size.
\begin{definition}
Let $\mathcal{I}=(\mathcal{S}, f)$ be a substream index.
The size of $\mathcal{I}$ is $size(\mathcal{I})=\max_{S\in \mathcal{S}}\{|S|\}$.
\end{definition}

Before discussing the query times, we bound the number of vertices that can be assigned to a substream.
\begin{lemma}\label{lemma:num_assigned_vertices}
The maximal number of vertices assigned to any substream $S\in\mathcal{S}$ is at most $2\cdot size(\mathcal{I})$
\end{lemma}
\begin{proof}
For $i\in [k]$, the number of vertices occurring in $S_i$ is $|\{u,v\mid (u,v,t,\lambda)\in S_i\}|\leq 2|S_i|$. Hence, the maximal number of vertices assigned to any substream $S_i$ is at most $2\cdot size(\mathcal{I})$.
\end{proof}

We now discuss the query time of the substream index.
\begin{theorem}\label[theorem]{theorem:runningtimes}
Let $\mathcal{I}$ be a substream index, $\Delta=size(\mathcal{I})$, and let $\delta$ be the maximal in-degree of a vertex in any of the substreams.
Given a query $(u,\tau)$, let $\sigma_+(u)$ the set of availability times of edges leaving the query vertex $u$, and $\rho=\min\{\delta, |\sigma_+(u)|\}$. 
Answering a fastest path query is possible in $\mathcal{O}(\Delta\log\rho)$, and if the transition times are equal for all edges, in $\mathcal{O}(\Delta)$. 
\end{theorem}
\begin{proof}
The result follows from the running times of the streaming algorithms~\cite{wu2014path} and \Cref{lemma:num_assigned_vertices}.
\end{proof}

\Cref{theorem:runningtimes} is the basis for the running time improvement of the temporal closeness computation.
Using the substream index, ranking all vertices according to their temporal closeness is possible with $n$ fastest path queries with a total running time in $\mathcal{O}(n \Delta)$.

As a trade-off, we need additional space for storing the substreams---for a temporal graph with $m$ edges and $n$ vertices, the space complexity of the substream index is in $\mathcal{O}(k\cdot m+n)$. 
Finally, it is noteworthy that a substream index can be used to output the distances as well as the corresponding paths,
and can be used to speed up other temporal distance queries, e.g., earliest arrival or latest departure time queries.

\subsection{Hardness of Finding a Minimal Index}
Unfortunately, deciding if there exists a substream index with a given size is NP-complete. %
\begin{theorem}\label{theorem:hardness2}
Given a temporal graph $\tg=(V,\tge)$, and $k,B\in \mathbb{N}$. Deciding if there exists a substream index $\mathcal{I}$ with $k+1$ substreams and $size(\mathcal{I})\leq B$ is NP-complete.
\end{theorem}
\begin{proof}
    We use a polynomial-time reduction from \textsc{UnaryBinPacking}, which is the following NP-complete problem \cite{garey1979computers}:\\
    \noindent\textbf{Given:} A set of $m$ items, each with positive size $w_i\in\mathbb{N}$ encoded in unary for $i\in [m]$, $k\in \mathbb{N}$,
    and $B\in\mathbb{N}$.
    \\\textbf{Question:} Is there a partition $I_1,\ldots, I_k$ of $\{1,\ldots,m\}$ such that $\max_{1\leq i\leq k}\{W_i\}\leq B$ where $W_i=\sum_{j\in I_i}w_j$?
    
    Given a substream index $\mathcal{I}$ for a temporal graph and $B\in\mathbb{N}$, we can verify in polynomial time if $size(\mathcal{I})\leq B$. Hence, the problem is in NP.
    We reduce \textsc{UnaryBinPacking} to the problem of deciding if there exists a substream index of size less or equal to $B$.
    Given an instance of \textsc{UnaryBinPacking} with $m$ items of sizes $w_i$ for $i\in[m]$, we construct in polynomial time a temporal graph $\tg=(V,\tge)$ that consists of $m$ vertices.
    More specifically, for each $i\in[m]$, we construct a vertex $v_i$ that has $w_i$ self-loops, such that
    each edge $e\in\tge$ has a unique availability time. 
    We show that if \textsc{UnaryBinPacking} has a \emph{yes} answer, then $size(\mathcal{I})\leq B$ and vice versa.
    \\$\Rightarrow:$ Let $I_1,\ldots, I_k$ be the partition of $[m]$ such that $\max_{1\leq i\leq k}\{W_i\}\leq B$. 
    For our \emph{Substream} index, we use $k$ substreams $S_i=\bigcup_{j \in I_i}\edgestream(v_j)$ for $i\in [k]$.
    Because $|\edgestream(v_j)|=w_j$ for $j\in [m]$ it follows that the size of the substream index 
    $size(\mathcal{I})=\max_{1\leq i\leq k}\{|S_i|\}=\max_{1\leq i\leq k}\left\{\textstyle\sum_{j\in I_i}w_j\right\}\leq B.$
    \\$\Leftarrow:$ Let $\mathcal{I}=(\mathcal{S}, f)$ with $\mathcal{S}=\{S_0,\ldots,S_k\}$ be the \emph{Substream} index with $size(\mathcal{I})=\max_{1\leq i\leq k}\{|S_i|\}\leq B$. 
    From the vertex mapping $f$, we construct the partition $I_1,\ldots, I_k$ of $[m]$ such that $\max_{1\leq i\leq k}\{W_i\}\leq B$ holds for the \textsc{UnaryBinPacking} instance.
    Let $I_i=\{j\in[m]\mid f(v_j)=i\}$ for $i\in[k]$.
    Finally, with $|\edgestream(v_j)|=w_j$, it follows that $\max_{1\leq i\leq k}\{\textstyle\sum_{j\in I_i}w_j\}\leq B$.
\end{proof}
\subsection{A Greedy Approximation}\label{sec:greedy}
We now introduce an efficient algorithm for computing a substream index with a bounded ratio between the size of the computed index and the optimal size.
\Cref{alg:greedy} shows the simple greedy algorithm for computing the substream index for a temporal graph in edge stream representation $S$. 
After initialization, \Cref{alg:greedy} runs $n$ iterations of the for loop in line \ref{alg:greedy:for}.
The vertices are processed in an arbitrary order $v_1,\ldots,v_n$.
In iteration $i$, the algorithm first computes the edge stream $\edgestream(v_i)$ using a single pass over the edge stream. 
After computing $\edgestream(v_i)$, it is added to one of the substreams, and the mapping $f(v_i)=j$ is updated.
In each round, $\edgestream(v_i)$ is added to one of the substreams $S_j\in\{S_1,\ldots,S_k\}$, such that the increase of the size of $S_j$ is minimal (line \ref{alg:greedy:find}.).
All vertices $v\in V$ for which $\edgestream(v)$ is empty are assigned to the empty substream $S_0$ (line \ref{alg:greedy:zero}). 
\begin{algorithm2e}[tb]\label[algorithm]{alg:greedy}
    \small
    \caption{Greedy index computation.}

    \Input{Temporal edge stream $S$, $k\in\mathbb{N}$.}
    \Output{Substream index $(\mathcal{S}, f)$}
    \BlankLine
    Initialize %
    $S_j=\emptyset$ for $1\leq j \leq k$, and $f(v)=0$ for $v\in V$ \\
    \For{$i=1$ to $n$}{\label{alg:greedy:for}
        Compute $\edgestream(v_i)$\\
        \lIf{$\edgestream(v_i)=\emptyset$}{$f(v_i)\gets 0$}\label{alg:greedy:zero}
        \Else{
            {Let $S_j$ such that $|S_j\cup \edgestream(v_i)|$ is minimal}\\\label{alg:greedy:find}
            $S_j \gets S_j \cup \edgestream(v_i)$\\
            $f(v_i)\gets j$
        }
    }
    \Return{$(\{S_1,\ldots,S_k\}, f)$}
\end{algorithm2e}

We use the following bounds to show the approximation ratio of the sizes of an optimal index and one constructed by \Cref{alg:greedy}. 

\begin{lemma}\label{lemma:bounds}
    Let $\tg$ be a temporal graph with $n^+$ non-sink vertices, $k\in \mathbb{N}$, and
    let $B\subseteq \{\edgestream(v)\mid v\in V\}$ such that $|B|=\left\lceil \frac{n^+}{k}\right\rceil$ and the union $\left|\bigcup_{\edgestream\in B}\edgestream\right|$ is maximal.
    \begin{compactenum}
        \item The size of an optimal index is $size(\textsc{Opt})\geq \frac{m}{k}$,
        \item and for greedy, $size(\textsc{Greedy})\leq \left|\bigcup_{\edgestream\in B}\edgestream\right|.$
    \end{compactenum}
\end{lemma}
\begin{proof}
    (1) Each edge $e\in \tge$ has to be in at least one substream. 
    The size of the substream index is minimized if the edges are distributed equally to all substreams---in this case, reassigning any edge from one of the $k$ substreams to another would lead to an increase of the size of the index. Hence, $size(\textsc{Opt})\geq\frac{m}{k}$.
    (2) 
    Assume in some iteration $i$ the edge stream $\edgestream(v_i)$ is added to a substream $S_j$ such that after the addition
    $ |S_j|>|\cup_{\edgestream\in B}\edgestream|$.
    \Cref{alg:greedy} chooses $S_j$ such that adding $\edgestream(v_i)$ to any other substream does not lead to a smaller value.
    Hence, for all $S_\ell$ with $1\leq \ell \leq k$ it is 
    $|S_\ell\cup \edgestream(v_i)|>|\cup_{\edgestream\in B}\edgestream|$.
    However, this leads to a contradiction to the maximal size of the sum over all substreams $\sum_{1\leq \ell\leq k}|S_\ell\cup \edgestream(v_i)|\leq k \cdot |\cup_{\edgestream\in B}\edgestream|$.
\end{proof}

\begin{theorem}\label[theorem]{theorem:approx}
    The ratio between the size of the greedy solution and the optimal solution for any valid input $S$ and $k$ is bounded by
        $\frac{size(\textsc{Greedy})}{size(\textsc{Opt})}\leq \frac{k}{\delta},\label[equation]{eq:approxratio}$
    with $\delta=\frac{m}{|\cup_{\edgestream\in B}\edgestream|}$ and $1\leq \delta$. %
\end{theorem}
\begin{proof}
    With the two bounds from \Cref{lemma:bounds}, follows 
    \begin{align*}
        \frac{size(\textsc{Greedy})}{size(\textsc{Opt})}&\leq \frac{|\bigcup_{\edgestream\in B}\edgestream|}{\frac{m}{k}} = \frac{k\cdot|\bigcup_{\edgestream\in B}\edgestream|}{m}=\frac{k}{\delta}
    \end{align*}
    And, because $\cup_{\edgestream\in B}\edgestream\subseteq\tge$ it follows $1\leq \delta$.
\end{proof}

We now discuss the running time of \Cref{alg:greedy}.
\begin{theorem}\label{theorem:runningtime_greedy}
    Given a temporal graph $S$ in edge stream representation, and  $k\in\mathbb{N}$, the running time of \Cref{alg:greedy} is in $\mathcal{O}(mnk)$.
\end{theorem}
\begin{proof}
    Initialization is done in $\mathcal{O}(n)$.
    The running time for computing all $\edgestream(v)$ of \Cref{alg:greedy} is in $\mathcal{O}(nm)$. For the assignment of the edge streams $\edgestream(v)$ to the substreams $S_1,\ldots, S_k$, the algorithm needs to compute the union $\edgestream(v_i)\cup S_j$ for each $j\in\{1,\ldots, k\}$. The sizes of $S_j$ and $\edgestream(v_i)$ are bounded by the number of edges $m$, therefore, union is possible in $\mathcal{O}(m)$ time.
    The algorithm needs $n\cdot(k+1)$ union operations, leading to a total running time of $\mathcal{O}(nmk)$.
\end{proof}

\subsection{Improving the Greedy Algorithm}\label{sec:substreamhash}
We improve the greedy algorithm presented in \Cref{sec:greedy} in three ways:
1) During the construction and queries, we skip edges that are too early in the temporal edge stream and do not need to be considered for answering a given query. 
2)~We use bottom-$h$\footnote{Usually called \emph{bottom-k} sketch. We use $h$ instead of $k$ because $k$ denotes the number of substreams.} sketches to avoid the costly union operations that we need to find the right substream to which we assign an edge stream.
3) We use parallelization and a batch-wise computation scheme to benefit from modern parallel processing capabilities.

Note that using only improvements 1) and 3), we would obtain a parallel greedy algorithm with the same approximation ratio as \Cref{alg:greedy}.  
However, improvement 2) leads to the loss of the approximation guarantee. In our experimental evaluation in \Cref{sec:experiments}, we will see that (i) the improved algorithm usually leads to indices that are not larger than the ones computed with \Cref{alg:greedy}, and (ii) the query times of the indices constructed with the improved algorithm are also faster for all data sets. 
In the following, we describe the improvements in detail.

\subsubsection{Edge Skipping}
The idea of edge skipping is to ignore all edges that have timestamps earlier than the availability time of the first edge leaving the query vertex $v$.
Let $S$ be a temporal edge stream with edges $e_1,\ldots, e_m$.
By definition, the edges are sorted in non-decreasing order of their availability times. 
The position of the first outgoing edge from vertex $v$ might be at a late position in the edge stream $S$.
For example, the first outgoing edge $e_p$ at $v$ could be at a position $p>m/2$. 
Therefore, if we know the position $p$ of the first edge, we can start the streaming algorithm at position $p$ and skip more than half of the edges in the run of the streaming algorithm.  
To exploit this idea, we store for each vertex $v\in V$ the first position $p$ in the edge stream $S$ of the first edge $e_p=(v,w,t,\lambda)$ that starts at vertex $v$. 
To compute the first positions, we first initialize an array of length $n$ in which we store the first positions of the earliest outgoing edges for each $v\in V$. We use a single pass over the edge stream to find these positions.
Hence, the array can be computed in $\mathcal{O}(n+m)$ running time, and it has a space complexity in $\mathcal{O}(n)$.
We use the edge skipping in two ways. First, it is used to speed up the computation of the edges streams $\edgestream(v)$ during the index construction. 
Secondly, we compute an array of starting positions for edge skipping for each of the final substreams in $\mathcal{S}$ to speed up the query times.

\subsubsection{Bottom-h Sketches}
The main drawback of \Cref{alg:greedy} is that it has to compute the union of $S_j\cup \edgestream(v_i)$ for all substreams $S_j$ for $1\leq j \leq k$ in order to determine the substream to which the edge stream $\edgestream(v_i)$ should be added. 
To avoid these expensive union computations, we reduce the sizes of $\edgestream(v)$ for $v\in V$ by using sketches of the edge streams, and estimate the \emph{Jaccard distance} between the sketches of the edge streams and substreams. 
For two sets $A$ and $B$, the Jaccard distance is defined as 
$J(A,B)=1-\frac{|A\cap B|}{|A\cup B|}$.
The Jaccard distance between two sets can be estimated using \emph{min-wise} hashing~\cite{broder1997resemblance}. 
The idea is to generate randomized sketches of sets that are too large to handle directly. After computing the sketches, further operations are done in the \emph{sketch space}. This way, it is possible to construct unions of sketches and estimate the Jaccard similarity between pairs of the original sets efficiently.

More specifically, let $A$ be a set of integers. A bottom-$h$ sketch $s(A)$ is generated by applying a permutation $\pi$ to the set $A$ and choosing $h$ smallest elements of the set $\{\pi(a)\mid a\in A\}$ ordered in non-decreasing value\footnote{We assume that $|A|\geq h$, otherwise we choose only $|A|$ elements.}. 
For two sets $A$ and $B$, we can obtain $s(A\cup B)$ by choosing $h$ smallest elements from $s(A)$ and $s(B)$.
This way, we obtain a sample of the union $A\cup B$ of size $h$.
Now, the subset $s(A\cup B) \cap s(A) \cap s(B)$ contains only the elements that are in the intersection of $A$, $B$, and the union sketch $s(A\cup B)$. We use the following result.
\begin{lemma}[\cite{broder1997resemblance}]
    The value
    \begin{center}
        $\hat{J}(A,B)=1-\frac{|s(A\cup B)\,\cap\,s(a)\,\cap\,s(b)|}{|s(A\cup B)|}$
    \end{center}
    is an unbiased estimator for the Jaccard distance. 
\end{lemma}

Using the estimated Jaccard distance between an edge stream $\edgestream(v)$ and a substream $S_j$, we decide if we should add $\edgestream(v)$ to $S_j$. If the estimated Jaccard distance is low, then we expect that adding $\edgestream(v)$ to $S_j$ does not lead to a significant increase in the size of $S_j$.

We now describe how we compute and use the sketches of the edge streams.
During the computation of $\edgestream(v_i)$, the algorithm iterates over the input stream $S$, starting from position $p$ determined by edge skipping array, and processes the edges $e_p,\ldots,e_\ell,\ldots,e_m$ in chronological order.
Let $e_\ell=(u,v,t,\lambda)$ be an edge that can be traversed, i.e., the arrival time at $u$ is smaller or equal to $t$.
We compute a bottom-$h$ sketch using the hashed position $\pi(\ell)$ of edge $e_\ell$ in the input stream.
Therefore, we compute a hash value for all edges that can be traversed, and we keep the $h$ smallest hashed values $\pi_1,\ldots, \pi_h$ as our sketch $s(\edgestream(v_i))=(\pi_1,\ldots,\pi_h)$, where the hash function $\pi$ is a permutation of $[m]$. 
Note that the position $\ell$ of edge $e_\ell$ in the edge stream $S$ is a unique identifier of $e_\ell$. 
Furthermore, each edge $e_\ell=(u,v,t,\lambda)$ represents a substream of $S$ consisting of all edges in $e=(x,y,t_e,\lambda_e)\in\edgestream(v)$ with availability time $t_e\geq t+\lambda$, i.e., the corresponding subgraph that is reachable after traversing $e$.

In the assignment phase (line \ref{alg:hashalg:for1}), \Cref{alg:improved} proceeds similarly to \Cref{alg:greedy} in a greedy fashion. However, we adapt the assignment objective such that it leads to improved substreams in terms of size and query times. 
To this end, we consider the number of vertices $I_j$ assigned to substream $S_j$. 

\begin{definition}
    Let $v\in V(S)$, and $I_j=|\{v\in V(S)\mid f(v)=S_j\}|$ the number of to $S_j$ assigned vertices. We define the \emph{ranking function} 
    $r:V(S)\times [k]\rightarrow [1,n]$ as 
    \begin{center}
        $r(v, j) = \frac{1}{2}(I_j +1)\cdot(\hat{J}(s(S_j), s(\edgestream(v))+1)$.
    \end{center}
\end{definition}

Using the ranking function, \Cref{alg:improved} decides to add the edge stream $\edgestream(v)$ to the substream $S_j$ if $r(v,j)$ is minimal for $1\leq j \leq k$ (line \ref{alg:hashalg:findi}). 
By additionally considering the number $I_j$ of vertices assigned to substream $S_j$, we optimize for small substreams $S_i$ and a vertex assignment such that vertices are assigned to smaller substreams $S_i$ rather than to larger ones. Note that if a vertex $u$ is assigned to a small substream, queries starting at $u$ can be answered fast.
If the ranking function $r(v,j)$ is close to one, not many vertices are assigned to $j$, or the estimated Jaccard distance between $\edgestream(v)$ and $S_j$ is small. On the other hand, if $r(v,j)$ is closer to $n$, the number of to $S_j$ assigned vertices is high, and/or the estimated Jaccard distance is high.
The intuition is that, even if we have a substream that contains a majority of edges, we want to assign the remaining vertices to substreams with a smaller size if possible. 
\subsubsection{Parallelization} %
\Cref{alg:improved} shows our improved \emph{parallel} greedy algorithm that has as input the temporal graph $S$, the number of substreams $k$, the hash-size $h$, and a batch-size $B>0$.
After the initialization and the computation of the edge skipping array,  
it processes the input graph in batches of size $B$ to allow a parallel computation of the edge stream assignment. 
The batch size determines how many vertices are processed in each iteration of the outer while-loop (line~\ref{alg:hashalg:while}).
For each batch of vertices, \Cref{alg:improved} runs three phases of computation.
In the first phase (line \ref{alg:hashalg:for0}), \Cref{alg:improved} first computes the edge streams $\edgestream(v_i)$ for all vertices $v_i$ that part of the current batch. The edge skipping array is used to find the first position $1\leq p\leq m$ of $v_i$ in $S$.
The second phase computes an assignment of the edge streams to the substreams using the bottom-$h$ sketches.
To this end, we keep the sketches $s(S_j)$ of the substreams stored as $C_j$ for each $j\in [k]$. 
After finding the right substream,
$C_j$, $I_j$ and $f(v_i)$ are updated accordingly.
The third phase (line \ref{alg:hashalg:for2}) constructs the substreams in parallel using the determined assignment of edge streams.
Finally, after all batches are processed, edge skipping arrays for each $S_i$ are computed in parallel (line \ref{alg:hashalg:timeskips}).
\begin{theorem}
    Given a temporal graph in edge stream representation with $m$ edges and $n\leq m$ vertices, and $B,h,k\in \mathbb{N}$ with 
    $h\geq1$, $k\geq2$, and $h\cdot k\leq m$. Then, the running time of \Cref{alg:improved} is in $\mathcal{O}(\frac{knm}{P})$ on a parallel machine\footnote{We consider the \emph{Concurrent Read Exclusive Write} (CREW) PRAM model.} with $P$ processors, for $P\leq k$ and $P\leq m$.
\end{theorem}
\begin{proof}
    Initialization is done in $\mathcal{O}(n)$.
    Computing the initial \emph{Time Skip} index for the input $S$ takes $\mathcal{O}(m)$ time.
    The algorithm iterates over $\lceil n/B \rceil$ batches.
    In one iteration of the while loop, the running time for the parallel computation of the edge sets $\mathcal{E}(v_i)$ is in $\mathcal{O}(B \cdot (n+m) / P)$.
    For the bottom-$h$ sketch, we use a sorted list to keep the smallest $h$ hash values of the edges. Updating the list is done in $\log h$.
    Finding the indices $i$ for the substreams $S_i$ in line \ref{alg:hashalg:findi} takes $\mathcal{O}(B\cdot ((kh/P)+\log P))$ time.
    Therefore, the total running time of the first two phases is 
    $\lceil n/B \rceil\cdot B \cdot (\mathcal{O}(m/P) + \mathcal{O}(kh/P+\log P))=\mathcal{O}(\frac{nm}{P}+n\log P)$.
    The total running time of the update phase is 
    $\lceil n/B \rceil \cdot (\mathcal{O}(\frac{k}{P}\cdot Bm)=\mathcal{O}(\frac{k}{P}nm)$.
    Computing the {edge skipping} arrays for $S_i$ with $i\in[k]$ in parallel takes $\mathcal{O}(\frac{k}{P}\cdot m)$ time.
\end{proof}

\begin{algorithm2e}[tb]\label[algorithm]{alg:improved}
    \small
    \caption{Parallel index computation.}
    \Input{Temporal edge stream $S$, $k, B\in\mathbb{N}$.} 
    \Output{Substream index $(\mathcal{S}, f)$}
    \BlankLine
    \SetNoFillComment

    Initialize in \ForParallel $I_i=0$, $S_i=\emptyset$, $C_i=\emptyset$ for $i\in [k]$, and $f(v)=0$ for $v\in V_S$ \\
    $start\gets 1$, $end\gets B$\\
    Compute edge skipping array for $S$\\
    \While{$start < n$}{ \label{alg:hashalg:while} 
        \tcc{Phase 1: compute streams \& sketches}
        \ForParallel \For{$i=start,\ldots, end$ } {\label{alg:hashalg:for0} 
            compute $\edgestream(v_i)$ and $s(\edgestream(v_i))$ using initial edge skipping array\\
        }
        \tcc{Phase 2: compute stream assignments}
        \For{$i=start,\ldots, end$ }{\label{alg:hashalg:for1} %
            \lIf{$\edgestream(v_i)=\emptyset$}{$f(v_i)\gets 0$}
            \Else{
                in \ForParallel find $j \in [k]$
                such that $r(v_i,j)$ is minimal\\\label{alg:hashalg:findi}
                $C_j \gets s(C_j \cup s(\edgestream(v_i)))$ \label{alg:hashalg:cjupdate}\\
                $f(v_i)\gets j$\\
                $I_j \gets I_j + 1$
            }
        }
        \tcc{Phase 3: updating substreams}
        \ForParallel \For{$j=1,\ldots, k$ } {\label{alg:hashalg:for2}
            \For{$v$ with $f(v)=j$} {
                $S_j \gets S_j \cup \edgestream(v)$\\
            }
        }
        $start\gets start + B$\\ 
        $end\gets \min(end + B,n)$\\
    }
    Compute in \ForParallel edge skipping array for $S_i$, $i\in[k]$ \label{alg:hashalg:timeskips}\\
    \Return{$(\{S_1,\ldots,S_k\}, f)$}
\end{algorithm2e}

\section{Experimental Results}\label{sec:experiments}
We implemented our algorithms in C++ using GNU CC Compiler 9.3.0. with the flag \texttt{--O3}, and we used OpenMP v4.5.
The source code is available at \url{https://gitlab.com/tgpublic/tgindex}.
The experiments ran on a computer cluster, where each experiment had an exclusive node with an Intel(R) Xeon(R) Gold 6130 CPU @ 2.10GHz and 192 GB of RAM. The time limit for each experiment was set to 48 hours.
Please refer to \Cref{appendix:experiments} for further experimental results.

\smallskip
\noindent\textbf{Algorithms:} We use the following algorithms.
\begin{compactitem}%
    \item \textsc{Greedy} is the implementation of \Cref{alg:greedy}.
    \item \textsc{SubStream} is the implementation \Cref{alg:improved}. 
    \item \textsc{TopChain} is the state-of-the-art index for single-source-single-destination (SSSD) temporal reachability queries~\cite{wu2016reachability}. We set the parameter $k=5$ as suggested in~\cite{wu2016reachability}.
    \item \textsc{OnePassFP} is temporal closeness algorithm based on the state-of-the-art SSAD edge stream algorithm for minimum duration distances~\cite{wu2014path}.
    \item \textsc{Top}-$\ell$ is the state-of-the-art temporal closeness algorithm~\cite{oettershagen2020efficient}. It computes the topmost $\ell$ closeness values and vertices exactly.  
    We set $\ell=100$.
\end{compactitem}
The C++ source codes of \textsc{TopChain}, \textsc{OnePassFp}, and \textsc{Top}-$\ell$ were provided by the corresponding authors and compiled using the same settings as our algorithms.

\smallskip
\noindent\textbf{Data sets:}~
We used the following real-world temporal graphs.
(1) \emph{Infectious:} a face-to-face human contact network~\cite{Isella2011}.
(2) \emph{AskUbuntu:} Interactions on the website \emph{Ask Ubuntu}~\cite{paranjape2017motifs}.
(3) \emph{Prosper:} A temporal network based on a personal loan website~\cite{redmond2013temporal}.
(4) \emph{Arxiv:} An author collaboration graph from the \emph{arXiv} website~\cite{leskovec2007graph}.
(5) \emph{Youtube:} A social network on the video platform \emph{Youtube}~\cite{MisloveMGDB07}.
(6) \emph{StackOverflow:} Interactions on the website \emph{StackOverflow}~\cite{paranjape2017motifs}.
\Cref{table:datasets_stats2} shows statistics of the data sets.

\begin{table}
    \centering
    \caption{Statistics of the data sets. } 
    \label{table:datasets_stats2}
    \resizebox{1\linewidth}{!}{ 	\renewcommand{\arraystretch}{0.9}
        \begin{tabular}{lrrrrr}
            \toprule
            \multirow{4}{0.5cm}{\vspace*{4pt}\textbf{Data~set}\vspace*{4pt}}&\multicolumn{5}{c}{\textbf{Properties}}\\
            \cmidrule{2-6} 
            \textbf{ }          & $|V|$          &  $|\tge|$      & $|\mathcal{T}(\tg)|$ &  avg.~$|\edgestream(v)|$  & $\max|\edgestream(v)|$  \\ \midrule
            \emph{Infectious}   & $10\,972$      & $415\,912$     &  $76\,943$     &  1\,100.1   & 9\,339 \\ 
            \emph{AskUbuntu}    & $159\,316$     & $964\,437$     &  $960\,866$    &  3\,050.8   & 117\,930 \\
            \emph{Prosper}      & $89\,269$      & $3\,394\,978$  &  $1\,259$      &  14\,979.4  & 205\,461 \\
            \emph{Arxiv}        & $28\,093$      & $4\,596\,803$  &  $2\,337$      &  260\,471.5 & 3\,860\,987 \\
            \emph{Youtube}      & $3\,223\,585$  & $9\,375\,374$  &  $203$         &  136\,682.2 & 4\,928\,847\\ 
            \emph{StackOverflow}& $2\,464\,606$ & $17\,823\,525$ & $16\,926\,738$  &  851\,232.8 & 11\,982\,619.0 \\
            \bottomrule
        \end{tabular}
    }
\vspace{-2mm}
\end{table}
\begin{table*}[t]\centering
    \caption{ Indexing times and sizes. We report the mean and standard deviation~over ten runs for \textsc{SubStream}.} 
    \begin{subtable}{\linewidth}
        \centering
        \caption{ Indexing times in (s). OOT---Out of time.} 
        \label{table:runningtimesindexing}
        \resizebox{0.95\textwidth}{!}{ 	\renewcommand{\arraystretch}{0.825}
            \begin{tabular}{lrrrr@{\hspace{5mm}}rrrr@{\hspace{5mm}}rr}
                \toprule
                \multirow{3}{0.5cm}{\textbf{Data~set}}&\multicolumn{4}{c}{\textsc{Greedy}}&\multicolumn{4}{c}{\textsc{SubStream}}&\multicolumn{1}{c}{\textbf{}}\\
                \textbf{ }         & $k=32$ & $k=64$ & $k=128$ & $k=256$ &  $k=32$ & $k=64$ & $k=128$ & $k=256$ & \textsc{TopChain} \\  
                \cmidrule(r){1-1} \cmidrule(lr){2-5}
                \cmidrule(lr){6-9}
                \cmidrule(lr){10-10}
                \emph{Infectious} & $4.4 $ & $3.0 $ & $2.4 $ & $2.6 $ & $0.75 \sz\pm 000.0$ & $0.78 \sz\pm 000.0$ & $0.79 \sz\pm 000.0$ & $0.78 \sz\pm 000.0$  & $0.43$  \\    
                \emph{AskUbuntu} & $201.9 $ & $216.7 $ & $211.7 $ & $229.3 $ & $6.43 \sz\pm 000.2$ & $6.33 \sz\pm 000.2$ & $6.09 \sz\pm 000.2$ & $6.10 \sz\pm 000.2$ & $0.53$  \\  
                \emph{Prosper} & $440.2 $ & $427.8 $ & $427.2 $ & $486.8 $ & $18.44 \sz\pm 000.7$ & $17.30 \sz\pm 000.4$ & $16.70 \sz\pm 000.3$ & $16.58 \sz\pm 000.3$ &  $1.98$  \\  
                \emph{Arxiv} & $2824.2 $ & $2895.9 $ & $3214.2 $ & $4223.6 $ & $28.29 \sz\pm 000.8$ & $25.52 \sz\pm 001.1$ & $22.56 \sz\pm 001.1$ & $21.38 \sz\pm 002.9$ &  $0.78$  \\  
                \emph{Youtube} & OOT & OOT & OOT & OOT & $5376.90 \sz\pm 573.1$ & $4813.78 \sz\pm 088.4$ & $4572.23 \sz\pm 068.6$ & $4360.68 \sz\pm 038.5$ & $12.46$ \\  
                \emph{StackOverflow} & OOT & OOT & OOT & OOT & $12467.22 \sz\pm 105.7$ & $11599.58 \sz\pm 235.1$ & $11307.12 \sz\pm 234.0$ & $11051.26 \sz\pm 184.7$ & $45.95$ \\
            \bottomrule
            \end{tabular}
        }	
    \end{subtable}
    \begin{subtable}{\linewidth}
        \centering
        \caption{ Index sizes in MiB. (``--'' indicates that the index is not available due to time out during construction).} 
        \label{table:index_sizes}
        \resizebox{0.95\textwidth}{!}{ 	\renewcommand{\arraystretch}{0.825}
            \begin{tabular}{lrrrrrrrrr}
                \toprule
                \multirow{3}{0.5cm}{\textbf{Data~set}}&\multicolumn{4}{c}{\textsc{Greedy}}&\multicolumn{4}{c}{\textsc{SubStream}}& \\
                \textbf{ }         & $k=32$ & $k=64$ & $k=128$ & $k=256$ &  $k=32$ & $k=64$ & $k=128$ & $k=256$ & \textsc{TopChain} \\  
                \cmidrule(r){1-1} \cmidrule(lr){2-5}
                \cmidrule(lr){6-9}
                \cmidrule(lr){10-10}
                
                \emph{Infectious} & $2.6 $ & $2.9 $ & $3.9 $ & $5.5 $ & $7.79 \sz\pm 000.4$ & $5.92 \sz\pm 000.3$ & $5.05 \sz\pm 000.2$ & $4.88 \sz\pm 000.1$ & $28.49$  \\
                \emph{AskUbuntu} & $15.2 $ & $29.1 $ & $56.3 $ & $106.9 $ & $9.36 \sz\pm 001.7$ & $10.41 \sz\pm 001.7$ & $9.87 \sz\pm 001.3$ & $12.67 \sz\pm 000.8$  & $20.41$  \\  
                \emph{Prosper} & $34.3 $ & $54.8 $ & $94.1 $ & $169.7 $ & $33.98 \sz\pm 001.4$ & $47.34 \sz\pm 002.3$ & $66.20 \sz\pm 005.7$ & $79.27 \sz\pm 007.9$  & $65.42$  \\
                \emph{Arxiv} & $457.2 $ & $883.8 $ & $1674.3 $ & $3123.6 $ & $295.98 \sz\pm 028.1$ & $459.46 \sz\pm 046.5$ & $572.59 \sz\pm 062.2$ & $721.02 \sz\pm 041.4$  & $18.86$  \\
                \emph{Youtube} & -- & -- & -- & -- & $541.53 \sz\pm 037.5$ & $863.72 \sz\pm 100.0$ & $1280.67 \sz\pm 147.9$ & $2038.11 \sz\pm 296.5$ & $351.03$ \\
                \emph{StackOverflow} & -- & -- & -- & -- & $907.51 \sz\pm 82.2$ & $1366.65 \sz\pm 111.7$ & $2201.15 \sz\pm 217.0$ & $3644.07 \sz\pm 375.5$  & $1501.00$ \\
                \bottomrule
            \end{tabular}
        }	
    \end{subtable}
\end{table*}

\subsection{Indexing Time and Index Size}
\label{sec:indexing}
For \textsc{SubStream}, we set the number of substreams to $k=2^i$ for $i\in\{5,\ldots,8\}$. 
We choose a sketch size of $h=8$ because in our experiments if showed a good trade-off between index construction times, query times,
and resulting index sizes. The construction time increases for larger values of $h$, however the gain in construction and query times diminished.
We set the batch size $B$ to $n$ for data sets with less than one million vertices and $2048$ otherwise. 
Furthermore, we used $32$ threads.
We report the indexing times in \Cref{table:runningtimesindexing} and the index sizes in \Cref{table:index_sizes}.
For \textsc{SubStream}, we run the indexing ten times and report the averages and standard deviations.

\smallskip
\noindent
\textbf{Indexing time:}
As expected, \textsc{Greedy} has high running times.
It has up to several orders of magnitude higher running times than the other indices, and for the two largest data sets, \emph{Youtube} and \emph{StackOverflow} the computations could not be finished in the given time limit of 48 hours.
\textsc{SubStream} improves the indexing time of \textsc{Greedy} immensely for all data sets.
However, \textsc{SubStream} has higher indexing times than \textsc{TopChain} for all data sets. 
The indexing time of \textsc{TopChain} is linear in the graph size, and the indexing for \textsc{SubStream} computes for each vertex $v\in V$ all reachable edges $\edgestream(v)$, hence higher running times and weaker scalability of \textsc{SubStream} are expected. 
However, the query times using \textsc{TopChain} for SSAD queries cannot compete with our indices and are, in most cases, orders of magnitude higher. The reason is that \textsc{TopChain} is designed for SSSD queries.

\smallskip
\noindent
\textbf{Index size:}
\Cref{table:index_sizes} shows that the index sizes of \textsc{Greedy} are only smaller than the ones of \textsc{SubStream} for the \emph{Infectious} data set. In all other cases, the \textsc{SubStream} size are (substantially) smaller.
Compared to \textsc{TopChain}, our \textsc{SubStream} can lead to larger sizes depending on $k$.
However, for \emph{Infectious} and \emph{AskUbuntu} the sizes of \textsc{SubStream} are smaller sizes for all $k$.
In general larger values of $k$ lead to larger indices, and shorter query times. %
Hence, \textsc{SubStream} provides a typical trade-off between index size and query time.
This is also demonstrated in \Cref{fig:tradeoff_stackoverflow}, which shows the trade-off for the two largest data sets, \emph{Youtube} and \emph{StackOverflow}, for increasing number of substreams $k$.
\begin{figure}
    \centering
    \includegraphics[width=0.35\linewidth]{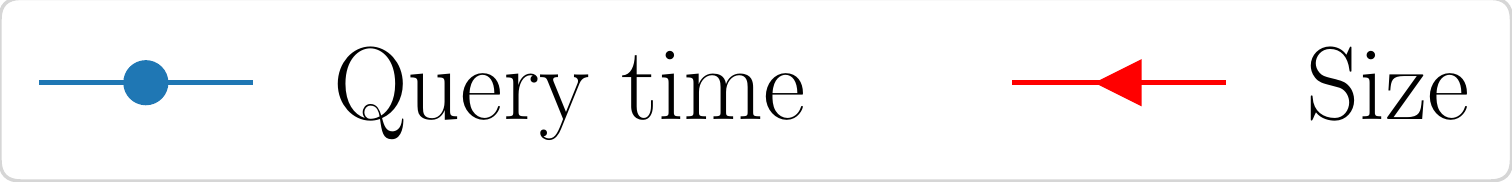}
    \includegraphics[width=0.45\linewidth]{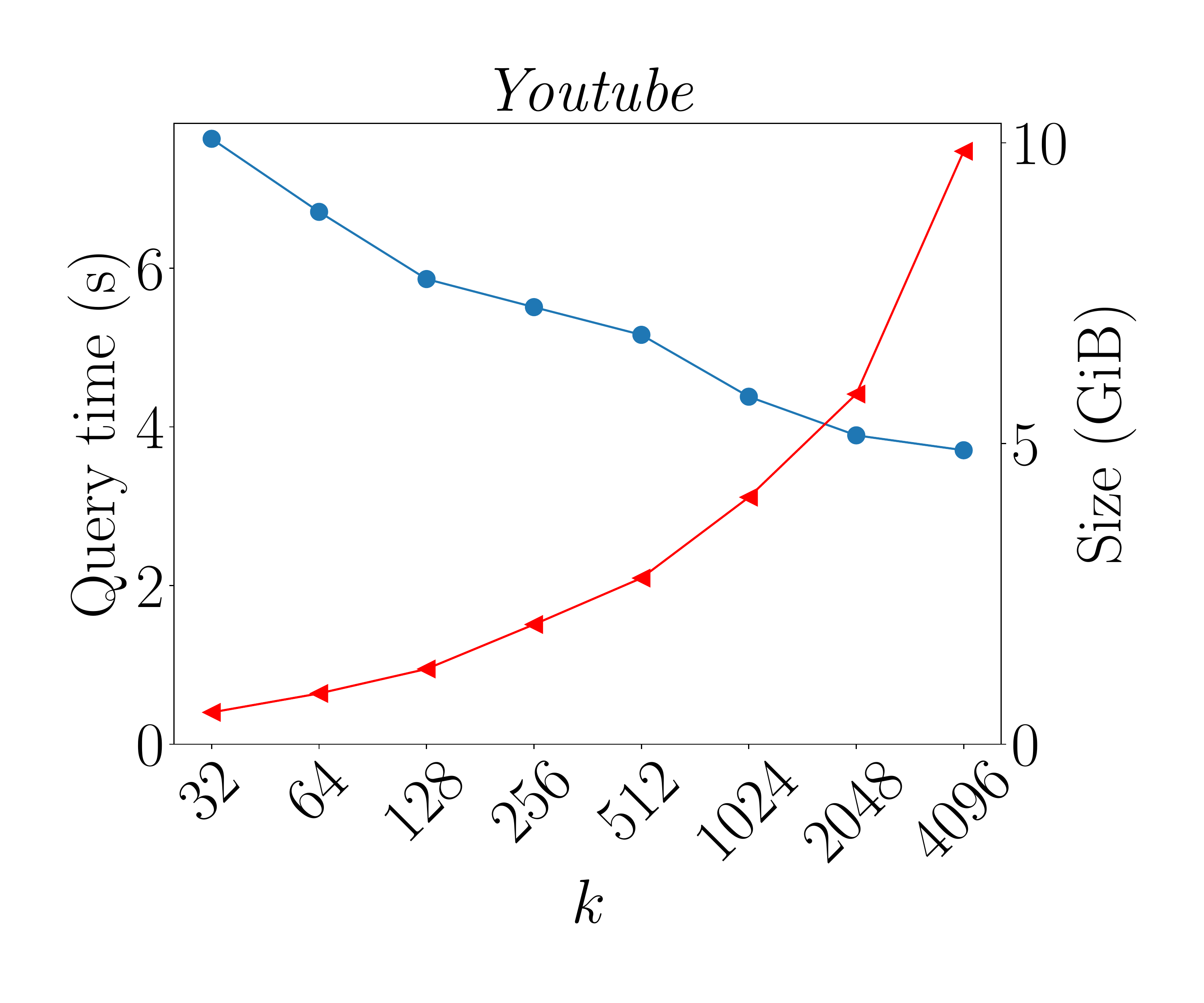}%
    \vspace{-2mm}%
    \includegraphics[width=0.45\linewidth]{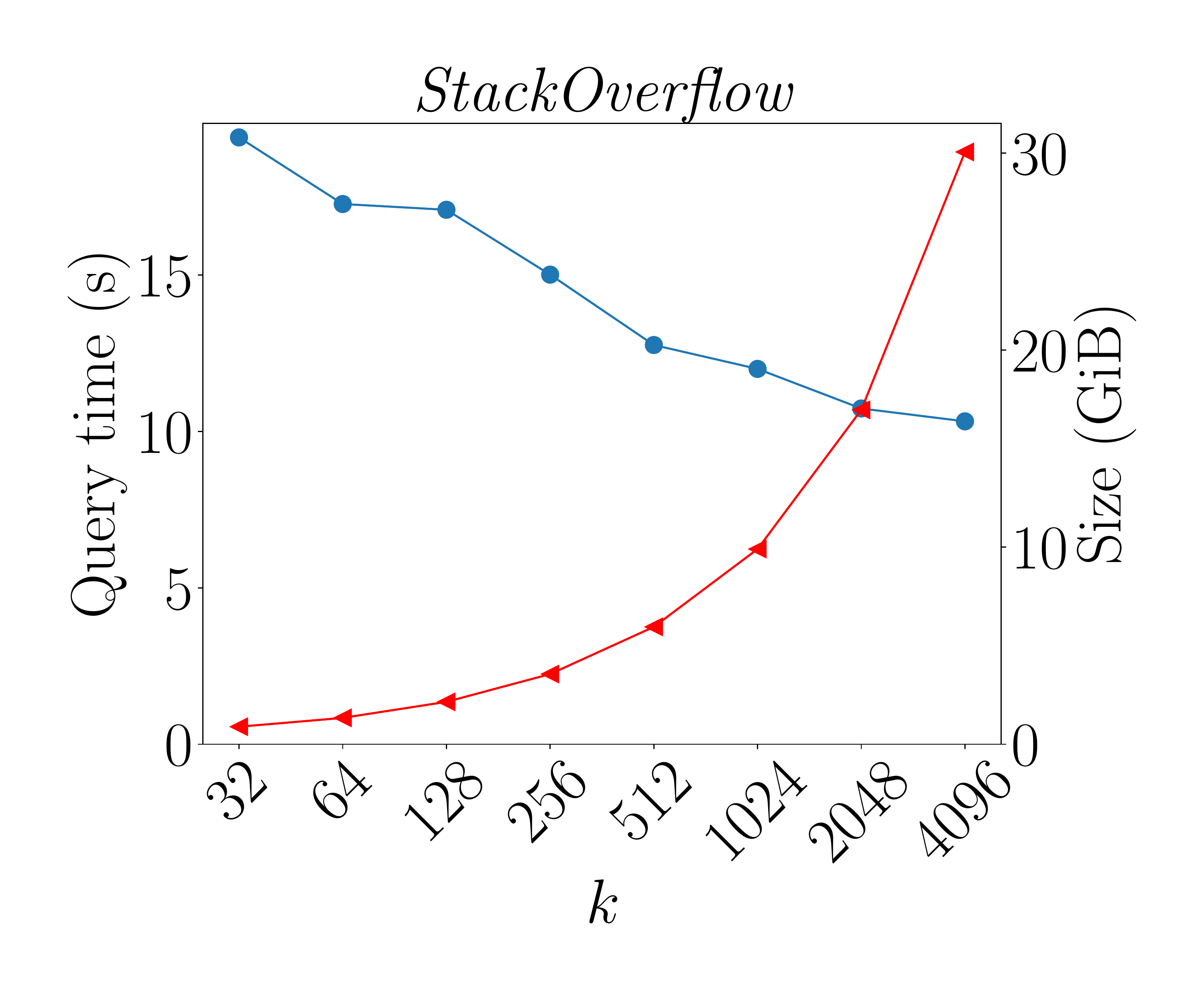}
    \caption{The trade-off between index size and query time for increasing $k$ for \emph{Youtube} and \emph{StackOverflow}.}
    \label{fig:tradeoff_stackoverflow}
    \vspace{-5mm}
\end{figure}

\subsection{Temporal Closeness Computation}\label{sec:usecase}

\begin{table}[t]
    \centering
    \caption{Running times for the temporal closeness in seconds and hours. OOT---{Out of time} after 7 days.} 
    \label{table:rtcloseness}
    \resizebox{1\linewidth}{!}{ 	\renewcommand{\arraystretch}{0.95}\setlength{\tabcolsep}{3mm}
        \begin{tabular}{lrrrr}%
            \toprule
            \textbf{ }           & \multicolumn{2}{c}{\textsc{SubStream}} & \multicolumn{2}{c}{\textsc{Baselines}}  \\  
            \textbf{\textbf{Data~set} } & $k=256$ &     $k=2048$   & \textsc{OnePassFp}  & \textsc{Top}-$100$  \\\midrule    
            \emph{Infectious}    &   1.67 s &   \textbf{  1.51} s  &   12.06 s     & 2.25 s   \\ 
            \emph{AskUbuntu}     & 102.95 s &   \textbf{102.46} s  &   229.73 s    & 132.53 s \\
            \emph{Prosper}       & 130.63 s &   \textbf{109.33} s  &   1\,665.20 s & 260.87 s \\ 
            \emph{Arxiv}         & 314.73 s &   \textbf{286.86} s  &   630.60 s    & 398.50 s \\
            \emph{Youtube}       &  63.82 h &   \textbf{ 59.72} h  &  145.98 h     & 81.21 h  \\ 
            \emph{StackOverflow} &  88.00 h &   \textbf{ 86.49} h  &  OOT          & 107.66 h \\
            \bottomrule
        \end{tabular}
    }	
    \vspace{-4mm}
\end{table}
We set $k\in\{256,2048\}$ and $h=8$ for \textsc{SubStream}. 
\Cref{table:rtcloseness} shows the running times. For the \textsc{SubStream} index, the running times \textbf{include the index construction times}. 
Our indices improve the running times for all data sets. Large improvements are gained for the \emph{Infectious} and \emph{Prosper} data sets compared to \textsc{OnePassFp} with speed-ups of eight and $15$, respectively. 
The speed-up for the other data sets is at least $2.2$ compared to \textsc{OnePassFp}. 
\textsc{OnePassFp} could not compute the ranking for \emph{StackOverflow} in the given time limit of seven days.
Compared to the \textsc{Top}-$100$ algorithm, the speed-up is between 1.2 (\emph{StackOverflow}) and 2.3 (\emph{Prosper}) with an average speed-up of $1.5$. Note that in contrast to \textsc{Top}-$100$, \textsc{SubStream} \emph{computes the complete ranking of all vertices}. 
As expected, the running time with $k=2048$ is shorter compared to $k=256$, even though the indexing times for $k=2048$ are slightly higher. %
The high speed-ups in the case of \emph{Infectious} and \emph{Prosper} can be explained by the small average and maximal sizes of $\edgestream(v)$ for both data sets (see~\Cref{table:datasets_stats2}), leading to an on average small maximum number of edges in the substreams of only $2.5\%$ and $5.8\%$ of the total number of edges. %
In conclusion, our \textsc{SubStream} index significantly improves the running times compared to the state-of-the-art algorithms for temporal closeness rankings, while computing the ranking of all vertices.

\section{Related Work}\label{sec:relatedwork}
\Cref{appenx:rw} provides further related work.
Recent and comprehensive introductions to temporal graphs are provided in, e.g.,~\cite{holme2015modern,wang2019time}. 
Wu et al.~\cite{wu2014path} introduce streaming algorithms for finding the fastest, shortest, latest departure, and earliest arrival paths. 
In~\cite{tang2013applications} and~\cite{tang2010analysing}, the authors compare temporal distance metrics and temporal centrality measures to their static counterparts. They reveal that the temporal versions for analyzing temporal graphs have advantages over static approaches on the aggregated graphs.
Variants of temporal closeness have been introduced in~\cite{pan2011path,magnien2015time,santoro2011time}.
Our work uses the harmonic temporal closeness definition from~\cite{oettershagen2020efficient}.
As far as we know, our work is the first one examining indices for \emph{SSAD temporal distance queries} and its application for temporal closeness computation.
Yu and Cheng~\cite{YuC10} give an overview of indexing techniques for reachability and distances in static graphs.
There are several works on SSSD time-dependent routing in transportation networks, e.g.,~\cite{BatzDSV09,Delling11}. %
Wang et al.~\cite{wang2015efficient} propose \emph{Timetable Labeling (TTL}), a labeling-based index for SSSD reachability queries based on hub labelings for temporal graphs.
In~\cite{wu2016reachability}, the authors introduce an index for SSSD reachability queries in temporal graphs called \textsc{TopChain}. The index uses a static representation of the temporal graph as a directed acyclic graph (DAG). 
On the DAG, a chain cover with labels at the vertices is computed. The labeling can be used to determine the reachability between vertices.
\textsc{TopChain} is faster than TTL and has shorter query times (see \cite{wu2016reachability}).
As far as we know, our work is the first one discussing indices for SSAD temporal distance queries.

\section{Conclusion and Future Work}
We introduced the \emph{Substream} index for speeding up temporal closeness computation.
Our index speeds up the vertex-ranking according to the temporal closeness up to one order of magnitude.
It can be extended to support efficiently dynamic updates for edge insertions or deletions. %
In future work, we want to further improve the indexing time using a distributed algorithm, and to explore further applications for our new index.

\subsection*{Acknowledgements}
This work is funded by the Deutsche Forschungsgemein\-schaft (DFG, German Research Foundation) under Germany's Excellence Strategy--EXC-2047/1--390685813.
This research has been funded by the Federal Ministry of Education and Research of Germany and the state of North-Rhine Westphalia as part of the Lamarr-Institute for Machine Learning and Artificial Intelligence, LAMARR22B.

\appendix
\section{Additional Experimental Results}\label{appendix:experiments}
In this section, we provide additional experimental results.
We use the following \textbf{additional} algorithms.
\begin{compactitem}
    \item \textsc{OnePass} is the SSAD edge stream algorithm for earliest arrival times~\cite{wu2014path}.
    \item \textsc{Dl} is a straight forward \emph{Dijkstra}-like approach using an adjacency lists representation of the temporal graph.
    \item \textsc{Xuan} is the algorithm for SSAD earliest-arrival paths from~\cite{xuan2003computing}. 
    \item \textsc{LabelFp} is a label setting algorithm for SSAD fastest paths using an adjacency list representation~\cite{oettershagen2020efficient,oettershagen2022computing}.
\end{compactitem}
\textsc{OnePass}, \textsc{OnePassFp}, and \textsc{LabelFp} were provided by the corresponding authors. 
We implemented \textsc{Xuan} using the graph data structure and algorithm described in~\cite{xuan2003computing}.

\subsection{Querying Time}\label{sec:querytime}
\begin{table*}[t]
    \centering
    \caption{Total querying time for 1000 queries in seconds. For \textsc{SubStream}, the mean and standard deviations over ten indices (``--'' indicates that the index is not available due to time out during construction).} 
    \label{table:reachquerysubstream}
    \begin{subtable}{1\linewidth}\centering
        \caption{Earliest arrival queries (OOT--out of time).} %
        \label{table:reachquerysubstreamea}
        \resizebox{1\linewidth}{!}{ 	\renewcommand{\arraystretch}{0.95}
            \begin{tabular}{lrrrrrrrrrrrrr}
                \toprule
                \multirow{4}{0.5cm}{\vspace*{4pt}\textbf{Data~set}\vspace*{4pt}}&\multicolumn{4}{c}{\textsc{Greedy}}&\multicolumn{4}{c}{\textsc{SubStream}}& \multicolumn{4}{c}{\textbf{Baselines}}\\
                \textbf{ }        &  $k=32$ & $k=64$ & $k=128$ & $k=256$ &  $k=32$ & $k=64$ & $k=128$ & $k=256$ &  \textsc{OnePass} &  \textsc{Dl} & \textsc{Xuan} & \textsc{TopChain} \\ 
                \cmidrule(r){1-1} \cmidrule(lr){2-5} \cmidrule(lr){6-9} \cmidrule(lr){10-10} \cmidrule(lr){11-11} \cmidrule(lr){12-12} \cmidrule(l){13-13}
                \emph{Infectious} & $0.014 $ & $0.013 $ & $0.011 $ & $\textbf{0.009} $ & $0.030 \sz\pm 0.0$ & $0.017 \sz\pm 0.0$ & $0.012 \sz\pm 0.0$ & $\textbf{0.009} \sz\pm 0.0$ &  $0.235$ & $0.023$ & $0.034$ & $0.144$\\ 
                \emph{AskUbuntu} & $0.069 $ & $0.073 $ & $0.079 $ & $0.078 $ & $0.063 \sz\pm 0.0$ & $0.057 \sz\pm 0.0$ & $0.054 \sz\pm 0.0$ & $\textbf{0.052} \sz\pm 0.0$ & $0.244$ & $0.482$ & $1.200$ & $10.643$\\ 
                \emph{Prosper} & $0.340 $ & $0.285 $ & $0.238 $ & $0.266 $ & $0.324 \sz\pm 0.0$ & $0.277 \sz\pm 0.0$ & $0.227 \sz\pm 0.0$ & $\textbf{0.198} \sz\pm 0.0$ & $2.828$ & $0.678$ & $4.440$ & $217.687$\\ 
                \emph{Arxiv} & $2.438 $ & $2.274 $ & $2.072 $ & $2.015 $ & $1.809 \sz\pm 0.1$ & $1.571 \sz\pm 0.1$ & $1.281 \sz\pm 0.1$ & $\textbf{1.118} \sz\pm 0.0$ & $4.418$ & $3.062$ & $43.937$ & $1526.040$\\ 
                \emph{Youtube} & -- & -- & -- & -- & $7.631 \sz\pm 0.6$ & $6.710 \sz\pm 0.6$ & $5.863 \sz\pm 0.3$ & $\textbf{5.509} \sz\pm 0.4$ & $7.368$ & $23.110$ &  $88.130$ & $45857.200$\\ 
                \emph{StackOverflow} & -- & -- & -- & -- & $19.405 \sz\pm 2.2$ & $17.275 \sz\pm 1.4$ & $17.091 \sz\pm 1.9$ & $\textbf{15.016} \sz\pm 1.4$ & $19.826$ &  $140.628$ & $534.967$ & OOT\\ 
                \bottomrule
            \end{tabular}
        }	
    \end{subtable}
    
    \begin{subtable}{.95\linewidth}\centering
        \caption{Minimum duration queries.} 
        \label{table:reachquerysubstreamfp}
        \resizebox{1\linewidth}{!}{ 	\renewcommand{\arraystretch}{0.95}
            \begin{tabular}{lrrrrrrrrrr}
                \toprule
                \multirow{4}{0.5cm}{\vspace*{4pt}\textbf{Data~set}\vspace*{4pt}}&\multicolumn{4}{c}{\textsc{Greedy}}&\multicolumn{4}{c}{\textsc{SubStream}}& \multicolumn{2}{c}{\textbf{Baselines}}\\
                \textbf{ }          &  $k=32$ & $k=64$ & $k=128$ & $k=256$ & $k=32$ & $k=64$ & $k=128$ & $k=256$ & \textsc{OnePassFp} &  \textsc{LabelFp}  \\
                \cmidrule(r){1-1}\cmidrule(lr){2-5}                  
                \cmidrule(lr){6-9} \cmidrule(lr){10-10}\cmidrule(l){11-11}
                \emph{Infectious} & $0.099$ & $0.102$ & $0.093$ & $\textbf{0.079}$ & $0.172 \sz\pm 0.0$ & $0.119 \sz\pm 0.0$ & $0.095 \sz\pm 0.0$ & ${0.086} \sz\pm 0.0$ & $0.908$ & $0.170$ \\ 
                \emph{AskUbuntu} & $0.842$ & $0.894$ & $0.889$ & $0.924$ & $0.838 \sz\pm 0.0$ & $0.777 \sz\pm 0.0$ & $0.765 \sz\pm 0.0$ & $\textbf{0.759} \sz\pm 0.0$  & $1.642$ & $1.272$ \\ 
                \emph{Prosper} & $2.079$ & $1.993$ & $1.425$ & $1.957$ & $2.253 \sz\pm 0.2$ & $1.988 \sz\pm 0.1$ & $1.840 \sz\pm 0.1$ & $\textbf{1.597} \sz\pm 0.1$  & $18.702$ & $3.005$ \\ 
                \emph{Arxiv} & $15.463$ & $15.385$ & $14.057$ & $14.546$ & $14.164 \sz\pm 0.8$ & $12.966 \sz\pm 0.3$ & $11.962 \sz\pm 0.4$ & $\textbf{11.365} \sz\pm 0.3$ & $21.023$ & $83.369$ \\ 
                \emph{Youtube} & -- & -- & -- & -- & $77.162 \sz\pm 6.5$ & $73.986 \sz\pm 8.8$ & $72.641 \sz\pm 0.8$ & $\textbf{70.787} \sz\pm 1.0$ & $163.542$ & $105.808$ \\ 
                \emph{StackOverflow} & -- & -- & -- & -- & $151.349 \sz\pm 6.2$ & $140.116 \sz\pm 2.9$ & $135.892 \sz\pm 2.3$ & $\textbf{134.563} \sz\pm 3.8$ & $303.240$ & $995.465$ \\                 
                \bottomrule
            \end{tabular}
        }	
    \end{subtable}%
\end{table*}

For each data set, we chose two random subsets $Q_1\subseteq V$ and $Q_2\subseteq V$ with $|Q_1|=|Q_2|=1000$. We run SSAD \emph{earliest arrival} from each vertex $u\in Q_1$, and \emph{minimum duration} queries from each vertex $u\in Q_2$. We used the same sets $Q_1$ and $Q_2$ for all algorithms.
Furthermore, we used the lifetime spanned by each temporal graph as the restrictive time interval. 
Because \textsc{TopChain} only supports SSSD queries, we added queries from each query vertex to all other vertices for \textsc{TopChain}. 
We report the average running times and standard deviations over ten separately constructed and evaluated \textsc{SubStream} indices.
Note that the queries are processed sequential for all indices and algorithms.

\subsubsection{Running times}
\Cref{table:reachquerysubstream} shows the total running times for the 1000 earliest arrival (\Cref{table:reachquerysubstreamea}) and minimum duration (\Cref{table:reachquerysubstreamfp}) queries.
Our indices perform best for all data sets and both query types.
For \textsc{TopChain}, we only report the running times of reachability queries in \Cref{table:reachquerysubstreamea} because the provided implementation does not support other types of queries.
The reported running times are lower bounds for the earliest arrival and minimum duration queries~(see~\cite{wu2016reachability}).
The query times of \textsc{TopChain} are up to several orders of magnitude higher than the running times of our indices for both earliest arrival and minimum duration queries. The reason is that \textsc{TopChain} is designed for SSSD queries. \textsc{TopChain} cannot answer the query for \emph{StackOverflow} in the time limit.
\textsc{SubStream} is the fastest for all data sets but in one case.
For \emph{Infectious} and $k=256$, \textsc{Greedy} is fastest.
In all other cases, \textsc{SubStream} answers queries in most cases slightly faster than 
\textsc{Greedy}.
The query times mostly decrease for increasing $k$ because the number of edges in each of the $k$ substreams is reduced; thus, fewer edges must be considered during the queries.
As expected, the running times increase with graph size in most cases.
Note that \textsc{Dl} and \textsc{Xuan} can be faster than \textsc{OnePass} if many vertices have limited reachability. The reason is that they stop processing when no further edge can be relaxed, and the priority queue is empty, where the streaming algorithms have to process the remaining stream until the end of the time interval. Similarly, \textsc{LabelFp} is faster than \textsc{OnePassFp} for some data sets.
Finally, \textsc{Dl} is faster than \textsc{Xuan} because the latter is primarily designed for temporal graphs with low dynamics~\cite{xuan2003computing}.

\subsection{Vertical Scalability and Batch Sizes}\label{sec:experiments:vs}

To evaluate the vertical scalability, we varied the number of threads in $\{1,2,4,8,16,32\}$.
Up to 16 threads, doubling the number of threads almost halves the running time. From 16 to 32 threads reduces the running time by more than 25\%.

Finally, we varied the batch size in $B\in\{1024,2048,4096,8192,$ $16384,|V|\}$ for \Cref{alg:improved}.
The parallel utilization of processing units was reduced for smaller batch sizes, and the indexing time increased. However, fewer edge streams need to be held in memory during the computation, and, therefore, the memory usage is less.
For larger batch sizes, the memory consumption increases while the indexing time decreases. 
Hence, the best indexing time can be achieved with large batch sizes. 
In the case of many vertices and large edge streams $\edgestream(v)$, a smaller batch size can reduce the amount of memory required for the computation.

\subsection{Indexing Times for $k=2048$}\label{appendix:indexingtimeslargek}
\Cref{table:largek} shows the running times of \Cref{alg:improved} for $k=2048$.
\begin{table}[t]
    \centering
    \caption{Indexing times in seconds for \textsc{SubStream} for $k=2048$.} 
    \label{table:largek}
    \begin{tabular}{lr}
        \hline
        \textbf{Data set} & {Indexing Time} \\\midrule
        \emph{Infectious} & 0.90 \\
        \emph{AskUbuntu} & 6.93 \\
        \emph{Prosper} & 18.47 \\
        \emph{Arxiv} & 18.96 \\
        \emph{Youtube} & 4\,186.58 \\
        \emph{StackOverflow} & 10\,543.16 \\
        \bottomrule
    \end{tabular}
\end{table}

\section{Additional Related Work}\label{appenx:rw}
\textbf{Temporal graphs and temporal paths.}
An early overview of dynamic graph models is given by Harary and Gupta~\cite{harary1997dynamic}.
More recent and comprehensive introductions to temporal graphs are provided in, e.g.,~\cite{holme2015modern,streamgraphs,santoro2011time,wang2019time}. 
There is an early work on temporal paths by Cooke and Halsey~\cite{cooke1966shortest}. 
Kempe et al.~\cite{KempeKK02} discuss time-respecting paths and related connectivity problems.
Xuan et al.~\cite{xuan2003computing} introduce algorithms for finding the shortest, fastest, and earliest arrival paths, which are generalizations of Dijkstra's shortest paths algorithm. 
In~\cite{HuangCW14}, variants of temporal graph traversals are defined. 
The authors of~\cite{tgenum} consider bicriteria temporal paths in weighted temporal graphs, where each edge has an additional cost value.
Wu et al.~\cite{wu2014path} introduce streaming algorithms for finding the fastest, shortest, latest departure, and earliest arrival paths. 
In~\cite{kanoulas2006finding}, the authors use the A* approach to find the set of all fastest $(u,v)$-paths in a road network and a given interval. 
\\\noindent\textbf{Indexing for static graphs.}
Yu and Cheng~\cite{YuC10} give an overview of indexing techniques for reachability and distances in static graphs.
Common approaches enrich the static graph with labels at the vertices that can be used to determine reachability or distances. Among them are hub/hop~\cite{CohenHKZ03,gavoille2004distance,ChengYLWY06,JinXRF09,peng2020answering,abraham2012hierarchical,li2017experimental}, landmark~\cite{akiba2013fast}, and interval labelings~\cite{shuo2015ailabel}, as well as tree~\cite{AgrawalBJ89,WangHYYY06,TrisslL07}, chain~\cite{ChenC08} and path covers~\cite{JinXRW08}.
Further approaches are specifically designed for road networks, e.g., are contraction hierarchies~\cite{GeisbergerSSD08,geisberger2012exact}, transit nodes~\cite{bast2006transit,bast2007fast,eisner2012transit}, and highway hierarchies~\cite{sanders2005highway,sanders2006engineering}.
\\\noindent\textbf{Indexing for temporal graphs.}
There are several works on SSSD time-dependent routing in transportation networks, e.g.,~\cite{BatzDSV09,BastCEGHRV10,Delling11,Geisberger10}.
Wang et al.~\cite{wang2015efficient} propose \emph{Timetable Labeling (TTL}), a labeling-based index for SSSD reachability queries that extends hub labelings for temporal graphs.
In~\cite{wu2016reachability}, the authors introduce an index for SSSD reachability queries in temporal graphs called \textsc{TopChain}. The index uses a static representation of the temporal graph as a directed acyclic graph (DAG). 
On the DAG, a chain cover with corresponding labels for all vertices is computed. The labeling can then be used to determine the reachability between vertices.
\textsc{TopChain} is faster than TTL and has shorter query times (see \cite{wu2016reachability}). We used \textsc{TopChain} as an SSSD baseline in our evaluation (see \Cref{sec:experiments}).  
The authors of~\cite{zhang2019efficient} present an index for reachability queries designed for distributed environments. 
It is similar to \textsc{TopChain} but forgoes the transformation into a DAG.
The authors of \cite{chen2021efficiently} use 2-hop labelings for indexing bipartite temporal graphs to answer reachability queries.
As far as we know, our work is the first one examining indices for SSAD temporal distance queries.

\end{document}